\documentclass[graybox]{icmcta4}

\usepackage[textsize=small,textwidth=0.3in]{todonotes}

\usepackage{amsmath}
\usepackage{mathptmx}
\usepackage{helvet}
\usepackage{courier}
\usepackage{type1cm}

\usepackage{makeidx}
\usepackage{graphicx}
\usepackage{multirow}
\usepackage[bottom]{footmisc}

\makeindex

\usepackage[latin1]{inputenc}
\usepackage{epstopdf}

\usepackage{amssymb}
\usepackage{enumerate}

\newcommand{\F}{\mathbb{F}}
\renewcommand{\P}{\mathbb{P}}
\newcommand{\Fq}{\mathbb{F}_q}
\newcommand{\ab}{\mathbf a}
\newcommand{\bb}{\mathbf b}

\newcommand{\X}{\mathcal X}

\newcommand{\Gm}{\mathbf{G}}

\newcommand{\C}[3]{\mathcal C_L(\mathcal{#1}, {#2},{#3})}
\newcommand{\GRS}[3]{\mathrm{GRS}_{#1}({#2},{#3})}

\renewcommand{\leq}{\leqslant}
\renewcommand{\geq}{\geqslant}

\begin{document}

\title*{Cryptanalysis of public-key cryptosystems that use subcodes of algebraic geometry codes}
\titlerunning{Cryptanalysis of PKC that use subcodes of AG codes}
\author{Alain Couvreur, Irene M\'arquez-Corbella and Ruud Pellikaan}
\institute{Alain Courvreur \at INRIA, SACLAY \& LIX, CNRS UMR 7161, \'Ecole Polytechnique 91128 Palaiseau Cedex, \email{alain.couvreur@lix.polytechnique.fr}
\and Irene M\'arquez-Corbella \at INRIA, SACLAY \& LIX, CNRS UMR 7161, \'Ecole Polytechnique 91128 Palaiseau Cedex,
\email{irene.marquez-corbella@inria.fr}
\and Ruud Pellikaan \at Departement of Mathematics and Computing Science, Eindhoven University of Technology, P.O. Box 513, 5600 MB Eindhoven
\email{g.r.pellikaan@tue.nl}}
%
%
\maketitle

\abstract{
  We give a polynomial time attack on the McEliece public key cryptosystem based on subcodes of algebraic geometry (AG) codes.
The proposed attack reposes on the distinguishability of such codes from
random codes using the Schur product.
Wieschebrink treated the genus zero case a few years ago but his approach
cannot be extent straightforwardly to other genera.
We address this problem by introducing and using a new notion, which we call the
{\em$t$--closure} of a code.
}




\section{Introduction}
\label{Section::1}

After the original proposal of code based encryption scheme due to McEliece
\cite{mceliece:1978} which was based on binary Goppa codes, several
alternative proposals aimed at reducing the key size by using codes with
a higher correction capacity. Among many others,
generalised Reed--Solomon (GRS) codes are proposed
in 1986 by Niederreiter
\cite{niederreiter:1986} but are subject to a key-recovery
polynomial time attack
discovered by Sidelnikov and Shestakov
\cite{sidelnikov:1992} in 1992.
To avoid this attack, Berger and Loidreau \cite{berger:2005}
proposed to replace
GRS codes by some random subcodes of small codimension.
This proposal has been broken by Wieschebrink \cite{wieschebrink:2010}
using Schur products of codes.

Another proposal was to use algebraic geometry (AG) codes, concatenated AG codes or their subfield subcodes \cite{janwa:1996}.
The case of AG codes of genus $1$ and $2$ has been broken by Faure and Minder
\cite{faure:2008}. Then, Marquez et. al. proved that the structure of a curve
can be recovered from the very knowledge of an AG code \cite{marquez:2014a, marquez:2014b} without leading to an efficient attack. Finally a polynomial
time attack of the scheme based on AG codes has been obtained by the authors
in \cite{ISIT14}. This attack consists in using the particular behaviour of
AG codes with respect to the Schur product to compute a filtration
of the public key by AG subcodes, which leads to the design of a polynomial
time decoding algorithm allowing encrypted message recovery.

The genus zero case and Berger Loidreau's proposal
raises a natural question
\textbf{what about using subcodes of AG codes?}
In this article we propose an attack of this scheme.
Compared to the genus zero case, Wieschebrink's attack cannot extend
straightforwardly and we need to introduce and use a new notion which we call
the {\em $t$--closure} of a code.
By this manner, we prove
subcodes of AG codes to be non secure when the subcode has
a small codimension.
It is worth noting that choosing a subcode of high codimension
instead of the code itself represents
a huge loss in terms of error correction capacity and hence
 is in general
a bad choice. For this reason, an attack on the small codimension codes
is of interest.

Finally, it hardly needs to be recalled that this result does not imply
the end of code-based cryptography since Goppa codes, alternant
codes and more generally subfield subcodes of AG codes still
resist to any known efficient attack.
Their resistance to the presented attack is discussed at the end of the article.

Due to space reasons, many proofs are omitted in this extended abstract.





\section{Notation and prerequisites}
\label{Section::nota}

\subsection{Curves and algebraic geometry codes}
The interested reader is referred to \cite{stichtenoth:2009, TVN} for further details on the notions introduced in the present subsection.
In this article, $\X$ denotes a smooth projective geometrically connected
curve  of genus $g$ over a finite field $\Fq$. We denote by
$P=(P_1, \ldots, P_n)$ an $n$-tuple of mutually distinct $\Fq$-rational
points of $\X$, by $D_{P}$ the divisor $D_{P}=P_1 + \cdots + P_n$ and by
$E$ an $\Fq$-divisor of degree $m\in \mathbb Z$ and support disjoint
from that of $D_{P}$.

 The function field of $\X$ is denoted by $\Fq(\X)$. Given an $\mathbb F_q$-divisor $E$ on $\X$, the corresponding Riemann-Roch space is denoted by $L(E)$.
The {\em algebraic geometry (AG) code} $\C{X}{P}{E}$ of length $n$ over $\mathbb F_q$ is the image of 
the evaluation map
$$
\mathrm{ev}_{P} : \left\{
  \begin{array}{ccc}
    L(E) & \longrightarrow & \Fq^n \\
    f    & \longmapsto     & (f(P_1), \ldots , f(P_n))
  \end{array}
\right.
$$
If $2g-2<m<n$, then by Riemann-Roch Theorem, $\C{X}{P}{E}$ has dimension $m+1-g$ and minimum distance at least $n-m$.

When the curve is the projective line $\P^1$, the corresponding codes are
the so-called {\em generalised Reed--Solomon} (GRS) codes defined as:
$$
\mathrm{GRS}_k(\ab, \bb) := \{(b_1 f(a_1), \ldots , b_n f(a_n))
~|~ f\in \Fq[x]_{<k}\}.
$$
where $\ab, \bb$ are two $n$--tuples in $\Fq^n$
such that the entries of $\ab$ are pairwise distinct and those of $\bb$ are all nonzero and $k<n$.

\begin{remark}
See \cite[Example 3.3]{hohpel} for a description
of GRS codes as AG codes.
\end{remark}

\subsection{Schur product}

Given two elements $\ab$ and $\bb$ in $\Fq^n$, the \emph{Schur product}
is the component wise multiplication: $\ab * \bb =
(a_1b_1, \ldots, a_nb_n)$
.
Let $\ab \in \Fq^n$, we set $\ab^0 :=(1,\ldots , 1)$
and by induction we define $\ab^{j+1}:=\ab * \ab^j$ for any
positive integer $j$. If all entries of $\bb$ are nonzero,
we define $\bb^{-1}:=(b_1^{-1}, \ldots , b_n^{-1})$ and thus, $\bb^{-j} = \left( \bb^j\right)^{-1}$ for any positive integer $j$.

For two codes $A, B \subseteq \Fq^n$, the code $A*B$ is defined by
$$A*B := \mathrm{Span}_{\Fq} \left\{ \ab * \bb  \mid \ab \in A \hbox{ and } \bb \in
B\right\}.$$
For $B=A$, then $A*A$ is denoted as $A^{(2)}$ and, we define $A^{(t)}$
by induction for any positive integer $t$.


\subsubsection{Application to Decoding, error correcting pairs and arrays}
The notion of \emph{error-correcting pair} (ECP) for a linear code
was introduced by Pellikaan \cite{pellikaan:1988,pellikaan:1992}
and independently by K\"otter \cite{koetter:1992}.
Broadly speaking, given a positive integer $t$, a $t$--ECP
for a linear code $\mathcal C \subseteq \mathbb F_q^n$
is a pair of linear codes $(A,B)$ in $\mathbb F_q^n$ satisfying
$A*B\subseteq \mathcal C^{\perp}$ together with several
inequalities relating $t$ and the dimensions and (dual) minimum distances
of $A, B, C$.
This data provides a decoding algorithm correcting up to $t$
errors in $O(n^3)$ operations in $\Fq$.
ECP's provide a unifying point of view for several
classical bounded distance decoding for algebraic and AG codes.
See \cite{marquez:2012b} for further details.

For an AG code, there always exists a $t$--ECP with
$t = \lfloor \frac{d^* - 1 - g}{2}\rfloor$, where $d^*$
denotes the {\em Goppa designed distance}
 (see \cite[Definition 2.2.4]{stichtenoth:2009}).
Thus, ECP's allow to correct up to half the designed distance minus
$g/2$. Filling this gap and correct up to half the designed distance
is possible thanks to more elaborate algorithms based on the so-called
{\em error correcting arrays}. See \cite{duursma:93,feng:93}
for further details.



\subsubsection{Distinguisher and Cryptanalysis}
Another and more recent application of the Schur product
concerns cryptanalysis of
code-based public key cryptosystems. In this context, the Schur product
is a very powerful operation which can help to distinguish some algebraic
codes such as AG codes from random ones.
The point is that evaluation codes do not behave
like random codes with respect to the Schur product: the square of an AG
code is very small compared to that of a random code of the same dimension. Thanks to this observation, Wieschebrink \cite{wieschebrink:2010}
gave an efficient attack of
Berger Loidreau's proposal \cite{berger:2005} based on subcodes
of GRS codes.

Recent attacks consist in pushing this argument forward and take advantage
to this distinguisher in order to compute a filtration of the public
code by a family of very particular subcodes.
This filtration method yields an alternative attack
on GRS codes \cite{CGGOT12}. Next it leads to a
key recovery attack on wild Goppa codes over quadratic extensions
in \cite{COT14}.
Finally in the case of AG codes, this approach lead to an attack
\cite{ISIT14} which consists in the computation of an ECP for the public code
without retrieving the structure of the curve, the points and the
divisor.

\section{The attack}
\label{Section::4}
Our public key is a non structured generator matrix $\Gm$ of a subcode $C$
of $\C{X}{P}{E}^{\bot}$ of dimension $l$, together with the error correcting capacity $t$. The goal of our attack is to recover the code
$\C{X}{P}{E}^{\bot}$ from the knowledge of $C$ and then use the attack
of \cite{ISIT14} which provides a $t$--ECP and hence
a decoding algorithm for $\C{X}{P}{E}$, which yields a fortiori
a decoding algorithm for $C$.


The genus zero case (i.e. the case of GRS codes) proposed in
\cite{berger:2005}
was broken by Wieschebrink \cite{wieschebrink:2010}
as follows:
\begin{itemize}
\item $C$ is the public key contained in some secret
$\GRS{k}{\ab}{\bb}$.
\item Compute $C^{(2)}$ which is, with a high probability,
equal to $\GRS{k}{\ab}{\bb}^{(2)}$,
which is itself equal to $\GRS{2k-1}{\ab}{\bb^2}$.
\item Apply Sidelnikov Shestakov attack \cite{sidelnikov:1992}
 to recover
$\ab$ and $\bb^2$, then find $\bb$.
\end{itemize}
Compared to Wieschebrink's approach, our difficulty is that
the attack \cite{ISIT14} is not a key-recovery
attack but a blind
construction of a decoding algorithm. For this reason,
even if $C^{(2)}$ provides probably
the code $\C{X}{P}{E}^{(2)}$, it is insufficient
for our purpose: we need to find $\C{X}{P}{E}$.
This is the reason why we introduce the notion of $t$--closures.

\subsection{The $t$-closure operation}
\label{Section::3}



\begin{definition}[$t$--closure]
Let $C \subset \Fq^n$ be a code and $t\geq 2$ be an integer.
The $t$-{\em closure} of $C$ is defined by
$$
\overline{C}^t = \left\{ \ab \in \Fq^n \mid \ab * C^{(t-1)} \subseteq C^{(t)}\right\}.
$$
The code $C$ is said to be $t$-{\em closed} if $\overline{C}^t = C$.
\end{definition}

\begin{proposition}
\label{Prop-Properties}
Let $C \in \Fq^n$, then for all $t\geq 2$,
$$
\overline{C}^t = {\left( C^{(t-1)} * {\left(C^{(t)}\right)}^{\bot} \right)}^{\bot}.
$$
\end{proposition}

\begin{proposition}
\label{Main-Prop}
Let $E$ be a divisor satisfying $\deg(E) \geq 2g+1$. Then:
\begin{enumerate}[(i)]
\item \label{Main-Prop:(1)} $\C{X}{P}{E}^{(t)} = \C{X}{P}{tE}$.
\item \label{Main-Prop:(2)}  $\overline{\C{X}{P}{E}}^t = \C{X}{P}{E}$
if $\deg(E)\leq \frac{n-2}{t}\cdot$
\end{enumerate}
\end{proposition}

\begin{proof}
(\ref{Main-Prop:(1)}) is proved in \cite{ISIT14} and is a consequence of
\cite{mumford:1970}.
For (\ref{Main-Prop:(2)}),
Proposition \ref{Prop-Properties}
shows that
\begin{equation}\label{eq:closure}
\overline{\C{X}{P}{E}}^t = \left(
\C{X}{P}{E}^{(t-1)} *
\left(\C{X}{P}{E}^{(t)}\right)^{\perp}
\right)^{\perp}.
\end{equation}
Moreover,
$\C{X}{P}{tE}^{\perp}  = \C{X}{P}{(tE)^{\perp}}$ where
$(tE)^{\perp} = D_P - tE + K$ for some canonical divisor $K$ on $\mathcal{X}$.
Thus, $\deg\left((tE)^{\perp}\right) = n-\deg(tE) +2g-2$.
Since, by assumption, $\deg(E) \leq \frac{n-2}{t}$
we have $\deg\left((tE)^{\perp}\right)\geq 2g$. Moreover, since
$\deg E \geq 2g+1$, then, thanks to (\ref{Main-Prop:(1)}), Equation
(\ref{eq:closure}) yields
$$\C{X}{P}{(t-1)E}* \C{X}{P}{tE}^{\perp} = \C{X}{P}{D_P-E+K} = \C{X}{P}{E}^{\bot}.
$$
\qed
\end{proof}
%
%
%
%

\begin{corollary}
\label{Main-Coro}
Let $E$ be a divisor and
$2g+1 \leq \deg(E)\leq \frac{n-2}{2}$. Then $\overline{\C{X}{P}{E}}^2 =
\C{X}{P}{E}$.
\end{corollary}


\begin{conjecture}\label{conj:squares}
If $2g+1 \leq \deg(E) \leq \frac{n-1}{2}$, let
$C$ be subcode of $\C{X}{P}{E}$ of dimension $l$ such that
and $2k +1-g\leq {l+1 \choose 2}$, where $k=\deg(E)+1-g$
is the dimension of $\C{X}{P}{E}$, then
the probability that $C^{(2)}$ is different from $\C{X}{P}{2E}$ tends to $0$ when $k$ tends to infinity.
\end{conjecture}

We give a proof along the lines of \cite[Remark 5]{marquez:2013} for the special case of subcodes of GRS codes.
Our experimental results are in good agreement with this conjecture
(see Table \ref{Table::1}).
The following corollary is central to our attack.
\begin{corollary}
\label{Main-Corollary}
If $2g+1 \leq \deg(E) \leq \frac{n-2}{2}$ and
$2k +1-g\leq {l+1 \choose 2}$ for $k=\deg(E)+1-g$, then the equality
$\overline{C}^{2} = \C{X}{P}{E}$ holds for random $l$-dimensional subcodes $C$ of $\C{X}{P}{E}$ with a probability tending to $0$ when $k$ tends to infinity.
\end{corollary}

\subsection{Principle of the attack}
The public key consists in
$
C \subseteq \C{X}{P}{E}^{\bot}$
and $t = \left\lfloor \frac{d^{*}-g-1}{2}\right\rfloor$.
Set $l:= \dim C$. First, let us assume moreover that
$$\begin{array}{cccc}
2g+1 \leq \deg(E) \leq \frac{n-1}{2}, &
k=\deg(E)+1-g & \hbox{ and } &
2k -1+g\leq {l+1 \choose 2}
\end{array}.$$

\begin{description}
\item \textbf{Step 1.} With a high probability, we may assume that $C^{(2)} = \C{X}{P}{2E}$ and hence
$\overline{C}^2 = \C{X}{P}{E}$ by Corollary \ref{Main-Corollary}.
Thus, compute $\overline{C}^2$ by solving a linear system or by
applying Proposition \ref{Prop-Properties}.
\item \textbf{Step 2.} Apply the polynomial time attack presented in \cite{ISIT14} to obtain an ECP, denoted by $(A,B)$, for $\C{X}{P}{E}$.
Which yields a decoding algorithm for $C$.
\medbreak
\item \textbf{Estimated complexity:} The computation
of a closure costs $O(n^4)$ operations in $\Fq$ and the
rest of the attack is in $O((\log (t+g)) n^4)$ (see \cite{ISIT14}
for further details).
\end{description}


In case $\deg(E)> \frac{n-1}{2}$,
then the attack can be applied to several shortenings of $C$
whose $2$--closures are computed separately and are then summed up to provide
$\C{X}{P}{E}$. This method is described and applied in \cite{COT14,ISIT14}.

\bigbreak

This attack has been implemented with MAGMA. To this end $L$ random subcodes of dimension $l$ from  Hermitian codes of parameters $[n,k]_q$ were created. It turned out that for all created subcodes a $t$-ECP could be reconstructed. {\tt Time} represents the average time of the attack obtained with an Intel $ \circledR$ CoreTM 2 Duo $2.8$ GHz.
The work factor $\mathbf{w}$ of an ISD attack is given.
These work factors have been computed thanks to Christiane Peter's Software \cite{peters:2010}.

\begin{table*}[h!]
\begin{tabular}{|c|c|c|c|c|c|c|c|c|}
\hline
$q$ & $n$ & $k$ & $t$ & {\tt Time}  & \hbox{key size} & $\mathbf w$ & $l$ & $L$ \\
\hline \hline
\multirow{3}{*}{$7^2$} & \multirow{3}{*}{$343$} & \multirow{3}{*}{$193$} & \multirow{3}{*}{$54$} &
\multirow{3}{*}{$80$ s}
& $83$ ko & $2^{30}$ & $50$ & $1000$ \\
& & & &  & $137$ ko & $2^{43}$ & $100$ & $1000$ \\
& & & &  & $163$ ko & $2^{62}$ & $150$ & $1000$ \\
\hline
\end{tabular}
\hfill
\begin{tabular}{|c|c|c|c|c|c|c|c|c|}
\hline
$q$ & $n$ & $k$ & $t$ & {\tt Time}  & \hbox{key size} & $\mathbf w$ & $l$ & $L$ \\
\hline \hline
\multirow{3}{*}{$9^2$} & \multirow{3}{*}{$729$} & \multirow{3}{*}{$521$} & \multirow{3}{*}{$19$} &
\multirow{3}{*}{$30$ min}
& $216$ ko & $2^{32}$ & $50$ & $500$ \\
& & & &  & $670$ ko  & $2^{121}$   & $200$ & $500$ \\
& & & &  & $835$ ko & $2^{178}$   & $400$ & $500$ \\
\hline
\end{tabular}
\caption{Running times of the attack over Hermitian codes}
\label{Table::1}
\end{table*}

\vspace{-1cm}

%
%
\subsection{Which codes are subject to this attack?}
Basically, the subcode $C \subseteq \C{X}{P}{E}$ should satisfy:
\begin{enumerate}[(i)]
\item\label{item:square} ${\dim C +1 \choose 2} \geq \dim \C{X}{P}{2E}$;
\item\label{item:dimC2E} $2g+1 \leq \deg E \leq \frac{n-2}{2}$;
\end{enumerate}
The left-hand inequality of (\ref{item:dimC2E}) is
in general satisfied. On the other hand,
as explained above, the right-hand inequality of (\ref{item:dimC2E})
can be relaxed by using a shortening trick.
Constraint (\ref{item:square}) is more central since a subcode
which does not satisfies it will probably behave like a random code and it
can be checked that a random code is in general $2$--closed. Thus,
computing the $2$--closure of such a subcode will not provide any
significant result. On the other hand, for an AG code of dimension $k$,
subcodes which do not satisfy (\ref{item:square}) have dimension
smaller than $\sqrt{2k}$ and choosing such very small subcodes
and decode them as subcodes of $\C{X}{P}{E}$ would represent
a big loss of efficiency. In addition, if these codes are too
small they can be subject to generic attacks like information set decoding.

\subsubsection{Subfield subcodes still resist}
Another class of subcodes which resist to this attack are the subcodes $C$
such that $\overline{C}^2 \varsubsetneq \C{X}{P}{E}$. It is rather difficult
to classify such subcodes but there is a very identifiable family: the subfield
subcodes. Let $\F$ be a proper subfield of $\F_q$
(here we assume $q$ to be non prime) and let $C := \C{X}{P}{E} \cap \F^n$ (and then apply a base field extension if one wants to have an $\Fq$--subcode).
The point is that $C^2 \subseteq (\C{X}{P}{E}^{(2)}) \cap \Fq^n$ and the
$2$-closure of $C$ will in general differ from $\C{X}{P}{E}$.
For this reason, subfield subcodes resist to this kind of attacks.
Notice that even in genus zero: subfield subcodes of GRS codes
still resist to filtration attacks unless for the cases presented
in \cite{COT14}.


\bibliography{29-ICMCTA4-Final}

\end{document}